\documentclass[11pt]{article}

\usepackage{theorem}
\usepackage{graphicx}
\usepackage{bm}
\usepackage{amsmath}
\usepackage{amssymb}
\usepackage{algorithm}
\usepackage{algorithmic}
\usepackage{tabularx}

\theoremstyle{plain}
\newtheorem{theorem}{Theorem}[section]
\newtheorem{lemma}[theorem]{Lemma}%[section]
\newtheorem{corollary}[theorem]{Corollary}%[section]
\newtheorem{proposition}[theorem]{Proposition}%[section]
\newsavebox{\ProofSym}
\savebox{\ProofSym}{%
  \begin{picture}(7,7)
    \put(0,0){\framebox(6,6){}}
    \put(0,2){\framebox(4,4){}}
  \end{picture}}
\newcommand{\eop}{\hfill \usebox{\ProofSym}} 
\newenvironment{proof}{\noindent {\it Proof.}}{\eop\par\vspace{0.3cm}}
\newenvironment{proof2}[1]{\noindent {\it Proof of #1.}}{\eop\par\vspace{0.3cm}}

\title{Covering Directed Graphs by In-trees}
\author{Naoyuki {\sc Kamiyama}\thanks{
Department of Architecture and Architectural Engineering,
Kyoto University, 
Kyotodaigaku-Katsura, Nishikyo-ku,
Kyoto, 615-8540, Japan.
E-mail$\colon$ {\ttfamily is.kamiyama@archi.kyoto-u.ac.jp} 
Supported by JSPS Research Fellowships for Young Scientists.}
\and 
Naoki {\sc Katoh}\thanks{Supported by the project {\it New Horizons in
Computing}, Grant-in-Aid for Scientific Research 
on Priority Areas, MEXT Japan.
Department of Architecture and Architectural Engineering,
Kyoto University, 
Kyotodaigaku-Katsura, Nishikyo-ku,
Kyoto, 615-8540, Japan.
E-mail$\colon$ {\ttfamily naoki@archi.kyoto-u.ac.jp} 
}
}
\date{\today}

\begin{document}

\maketitle

\begin{abstract}
\noindent
Given a directed graph $D=(V,A)$ with a set of $d$ specified vertices
$S=\{s_1,\ldots,s_d\}\subseteq V$ and a function $f\colon
S \to \mathbb{Z}_+$ where $\mathbb{Z}_+$ denotes the set of non-negative integers, 
we consider the problem which asks whether there exist $\sum_{i=1}^d f(s_i)$ in-trees denoted by 
$T_{i,1},T_{i,2},\ldots, T_{i,f(s_i)}$ for every $i=1,\ldots,d$ such that
$T_{i,1},\ldots,T_{i,f(s_i)}$ are rooted at $s_i$, each $T_{i,j}$
spans vertices from which $s_i$ is reachable and the union of all arc sets of
$T_{i,j}$ for $i=1,\ldots,d$ and $j=1,\ldots,f(s_i)$ covers $A$. In this paper, we prove that such set of 
in-trees covering $A$ can be found by using an algorithm for the weighted matroid intersection problem
in time bounded by a polynomial in $\sum_{i=1}^df(s_i)$ and the size of $D$. 
Furthermore, for the case where $D$ is acyclic, 
we present another characterization of the existence of in-trees covering $A$,  
and then we prove that in-trees covering $A$ can be 
computed more efficiently than the general case by finding maximum matchings in a series of bipartite graphs.
\end{abstract}

\section{Introduction}

The problem for covering a graph by subgraphs with specified properties (for example, trees or paths) is
very important from practical and theoretical viewpoints and have
been extensively studied. For example, Nagamochi and Okada~\cite{NO07}
studied the problem for covering a set of vertices of a given undirected
tree by subtrees, and Arkin et al.~\cite{AHL06} studied the problem  
for covering a set of vertices or edges of a given undirected graph by
subtrees or paths. These results were motivated by vehicle routing problems. 
Moreover, Even et al.~\cite{EGKRS04} studied the covering problem
motivated by nurse station location problems. 

This paper studies the problem for covering a directed graph by rooted trees which is 
motivated by the following evacuation planning problem. 
Given a directed graph which models a city, 
vertices model intersections and buildings,
and arcs model roads connecting these intersections and buildings.
People exist not only at vertices but also along arcs. 
Suppose we have to give several evacuation instructions for
evacuating all people to some safety place. In order to avoid disorderly
confusion, it is desirable that one evacuation instruction gives a single evacuation path
for each person and these paths do not cross each other. 
Thus, we want each evacuation instruction to become an
in-tree rooted at some safety place. Moreover, the number of instructions
for each safety place is bounded in proportion to a size of each safety place. 

The above evacuation planning problem is formulated as the following
covering problem defined on a directed graph. 
We are given a directed graph $D=(V,A,S,f)$ 
which consists of a vertex set $V$, an arc set $A$, a set of $d$ specified vertices
$S=\{s_1,\ldots,s_d\}\subseteq V$ and a function $f\colon S\to
\mathbb{Z}_+$ where $\mathbb{Z}_+$ denotes the set of non-negative integers. 
In the above evacuation planning problem, $S$ corresponds to a set of safety
places, and $f(s_i)$ represents the upper bound of
the number of evacuation instructions for $s_i \in S$. 
For each $i=1,\ldots,d$, we define $V^i_D \subseteq V$ as the
set of vertices in $V$ from which $s_i$ is reachable in $D$, and we define an in-tree
rooted at $s_i$ which spans $V^i_D$ as a {\it $(D,s_i)$-in-tree}.
We define a set $\mathcal{T}$ of $\sum_{i=1}^df(s_i)$ subgraphs of $D$ 
as a {\it $D$-canonical set of in-trees}
if $\mathcal{T}$ contains exactly $f(s_i)$ $(D,s_i)$-in-trees for every $i=1,\ldots,d$. 
If every two distinct in-trees of a $D$-canonical set $\mathcal{T}$ 
of in-trees are arc-disjoint, we call $\mathcal{T}$ a {\it 
$D$-canonical set of arc-disjoint in-trees}. 
Furthermore, if the union of arc sets of all
in-trees of a $D$-canonical set $\mathcal{T}$ of in-trees 
is equal to $A$, we say that $\mathcal{T}$ {\it covers} $A$. 

Four in-trees illustrated in Figure~\ref{fig:covering2}
compose a $D$-canonical set $\mathcal{T}$ of in-trees which covers the arc set of a
directed graph  $D=(V,A,S,f)$ illustrated in Figure~\ref{fig:covering1}(a) where 
$S=\{s_1,s_2,s_3\}$, $f(s_1)=2$, $f(s_2)=1$ and $f(s_3)=1$. However, $\mathcal{T}$ 
is not a $D$-canonical set of arc-disjoint in-trees. 
\begin{figure}[h]
\begin{minipage}{0.5\hsize}
\begin{center}
\includegraphics[width=5cm]{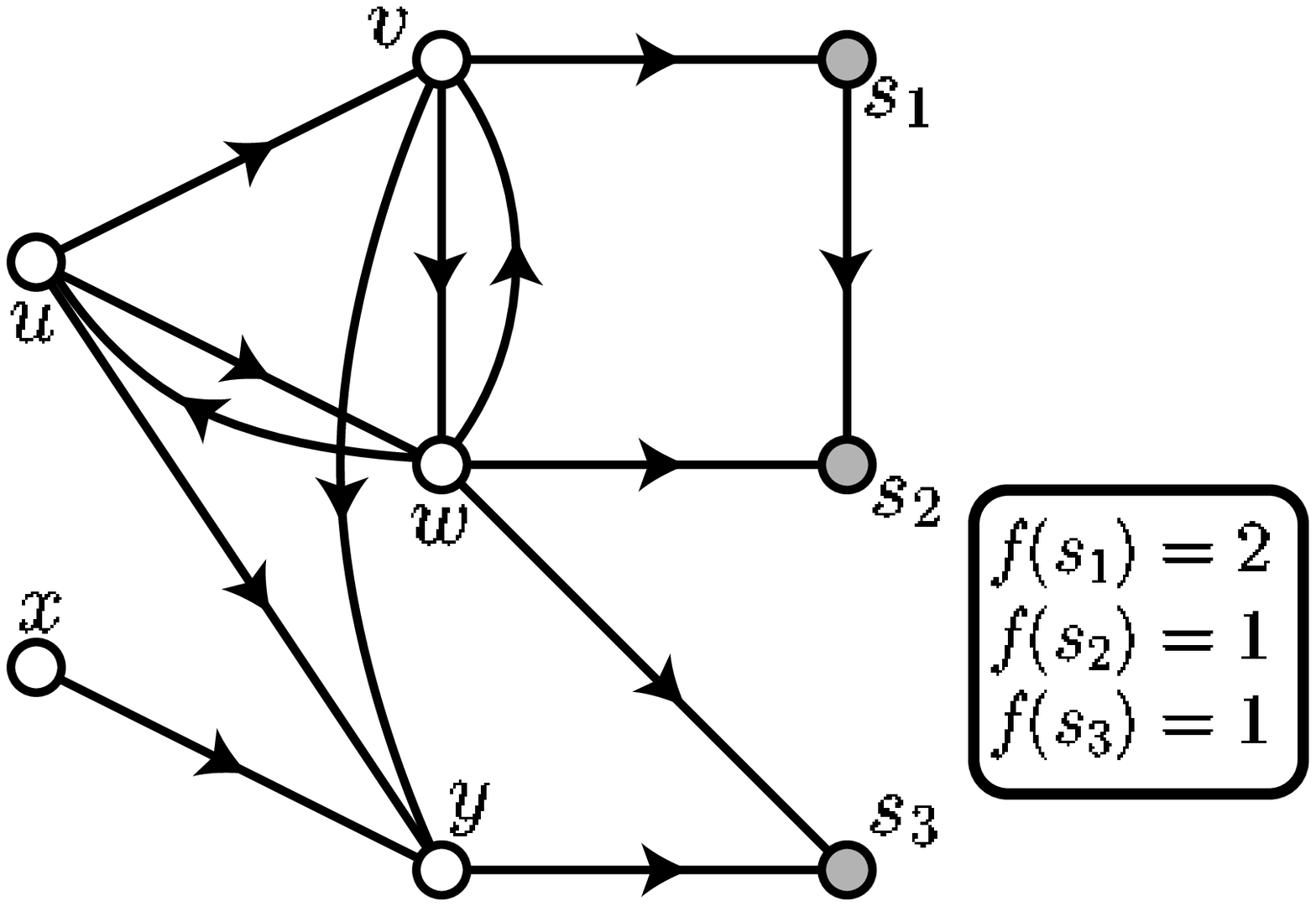}
\par(a)
\end{center}
\end{minipage}
\begin{minipage}{0.5\hsize}
\begin{center}
\includegraphics[width=5cm]{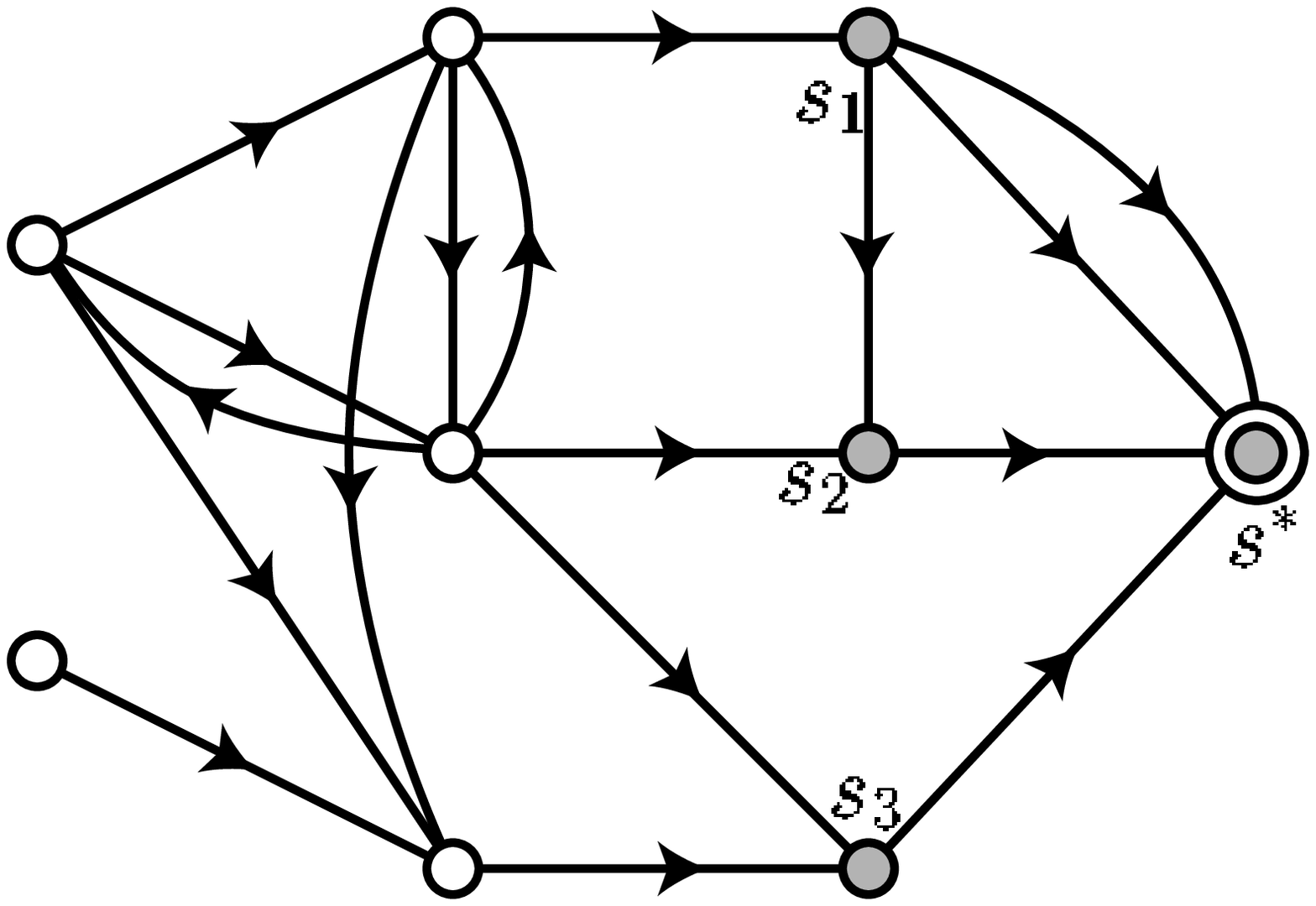}
\par(b)
\end{center}
\end{minipage}
\caption{\small (a) Directed graph  $D$. (b) Transformed graph $D^{\ast}$.}
\label{fig:covering1}
\end{figure}
\begin{figure}[h]
\begin{minipage}{0.245\hsize}
\begin{center}
\includegraphics[width=3.5cm]{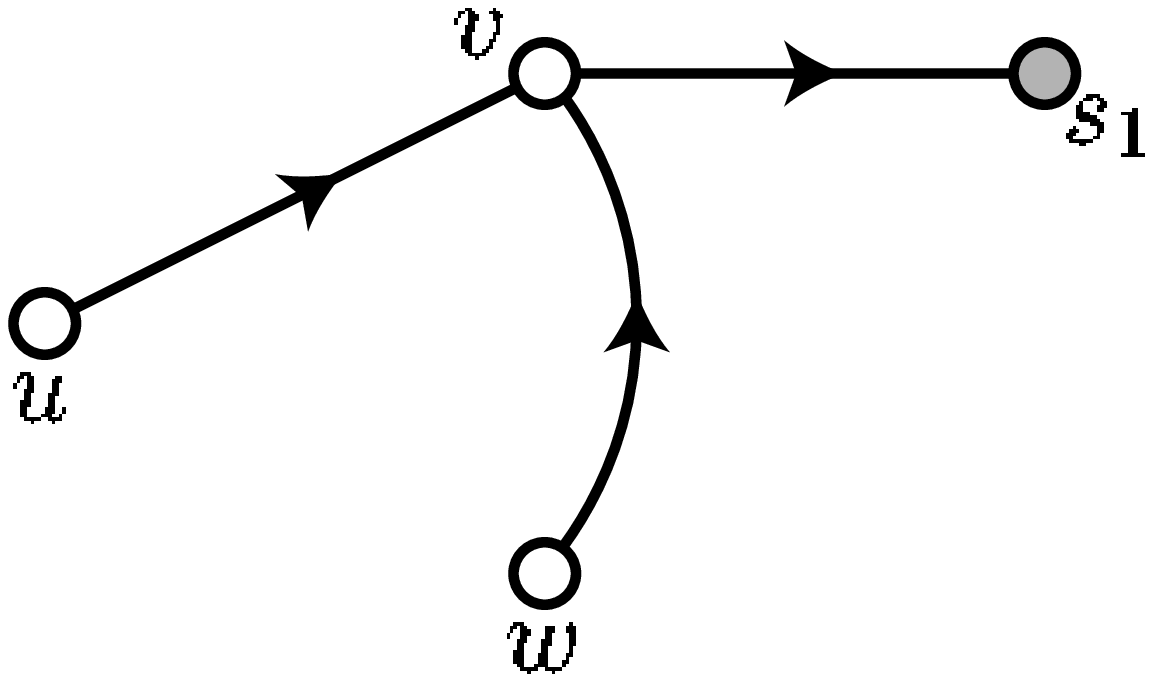}
\par(a)
\end{center}
\end{minipage}
\begin{minipage}{0.245\hsize}
\begin{center}
\includegraphics[width=3.5cm]{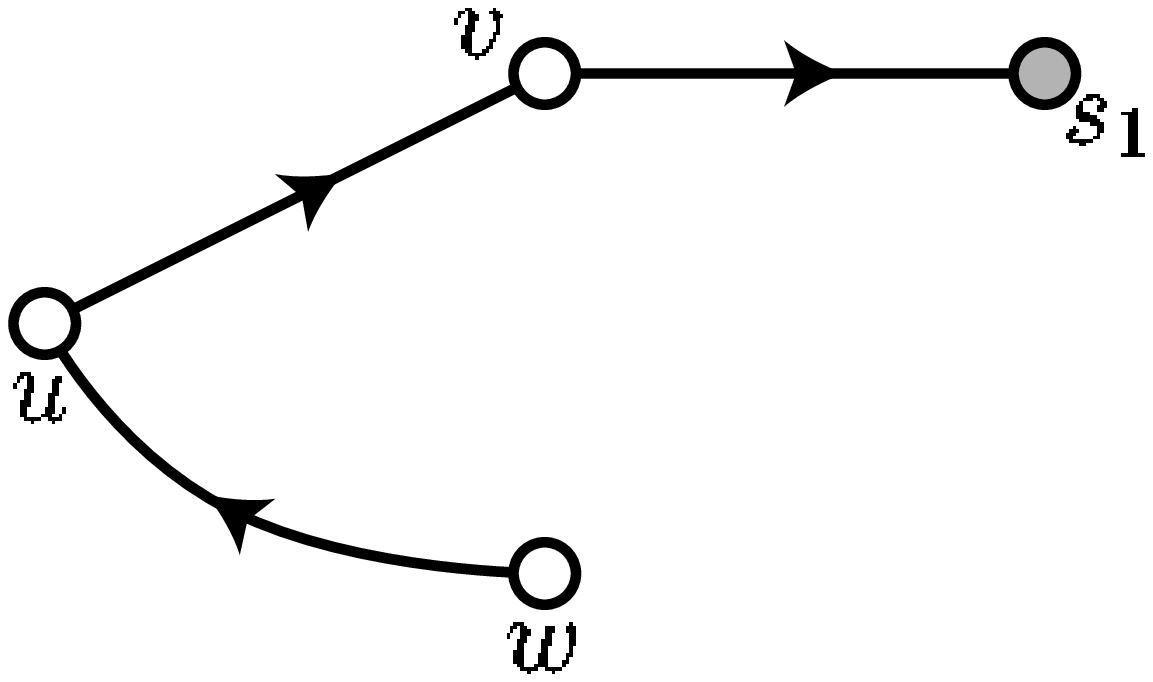}
\par(b)
\end{center}
\end{minipage}
\begin{minipage}{0.245\hsize}
\begin{center}
\includegraphics[width=3.5cm]{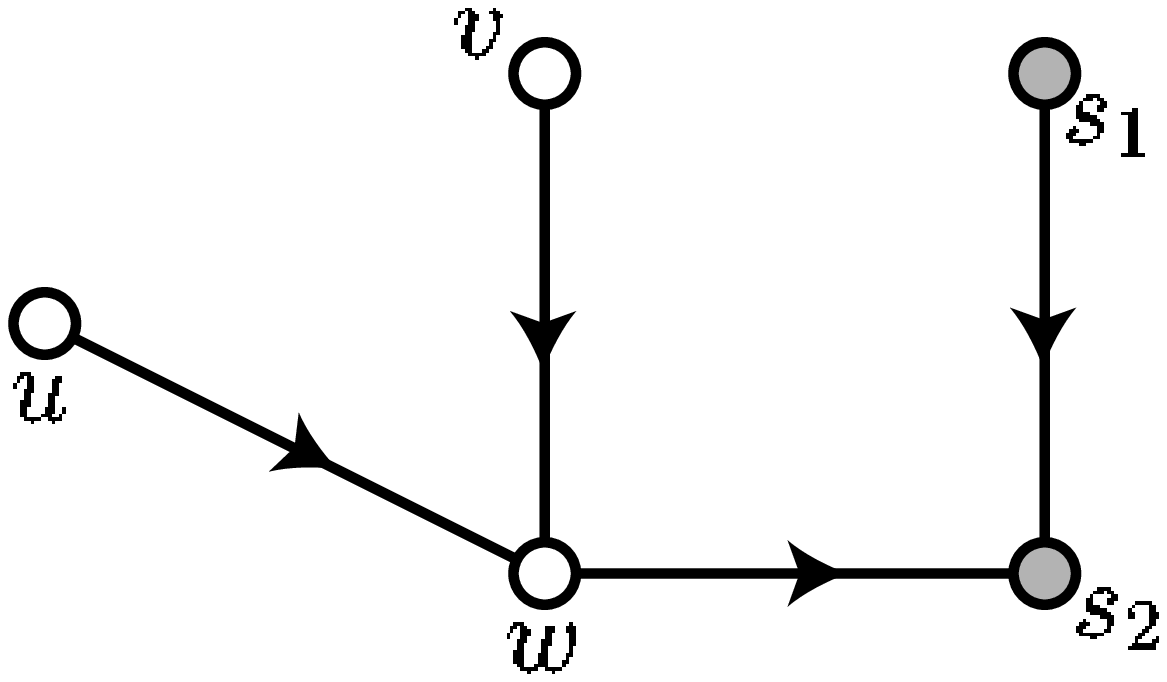}
\par(c)
\end{center}
\end{minipage}
\begin{minipage}{0.245\hsize}
\begin{center}
\includegraphics[width=3.5cm]{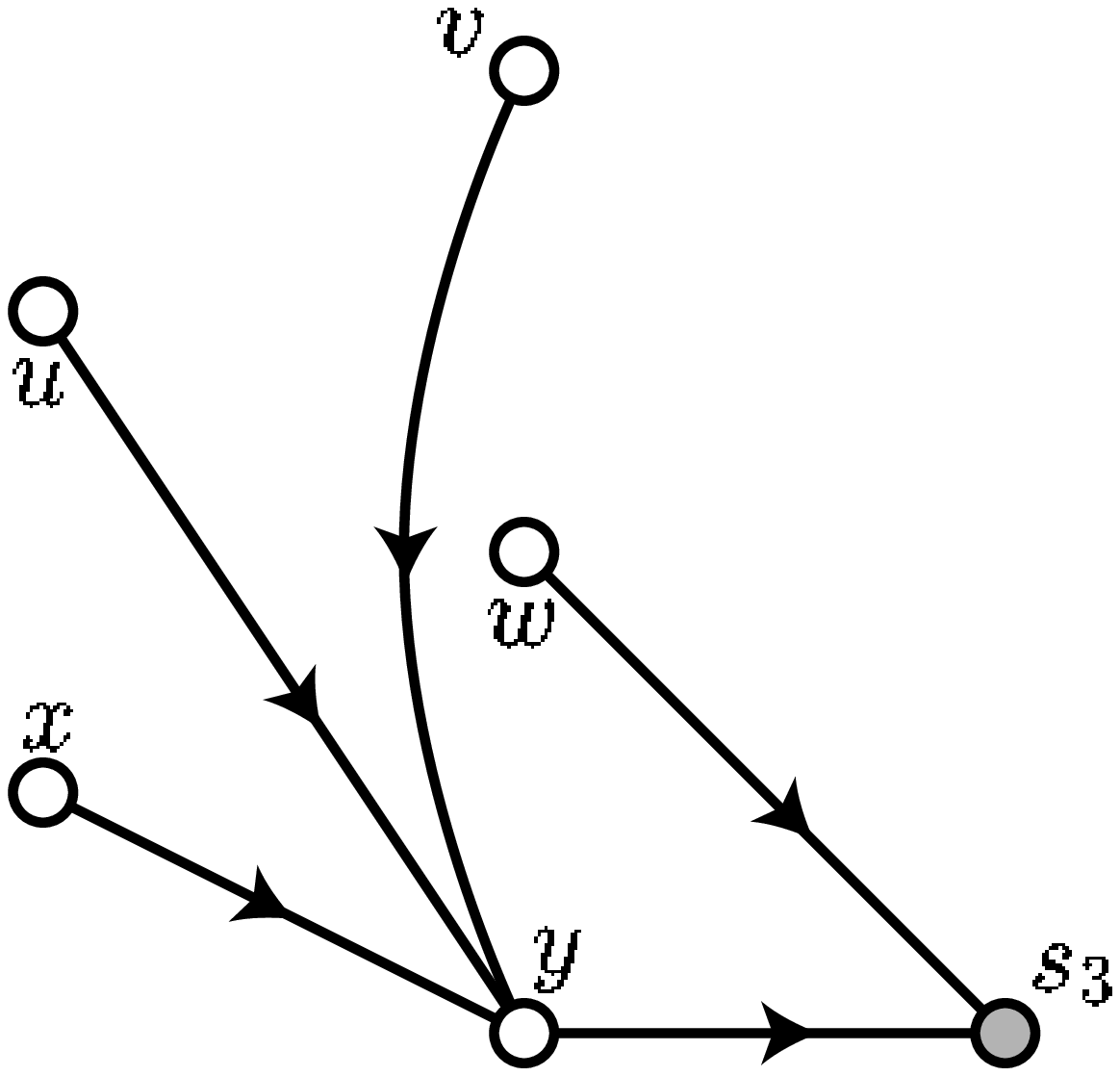}
\par(d)
\end{center}
\end{minipage}
\caption{\small (a) $(D,s_1)$-in-tree. (b) $(D,s_1)$-in-tree. (c) $(D,s_2)$-in-tree. (d) $(D,s_3)$-in-tree.}
\label{fig:covering2}
\end{figure}

We will study the problem for 
{\it covering directed graphs by in-trees} (in short $\mbox{CDGI}$), 
and we will present characterizations for a directed graph $D=(V,A,S,f)$ 
for which there exists a feasible solution of $\mbox{CDGI}(D)$, and 
a polynomial time algorithm for $\mbox{CDGI}(D)$. 
\begin{center}
\begin{tabularx}{150mm}{rX}
\hline
{\bf Problem$\colon$} & $\mbox{CDGI}(D)$ 
\\
\hline 
{\bf Input$\colon$} & a directed graph $D$; \\
{\bf Output$\colon$} &
a $D$-canonical set of in-trees which covers the arc set of $D$, if one exists.\\
\hline
\end{tabularx}
\end{center}
A special class of the problem $\mbox{CDGI}(D)$ in which $S$ consists of a single vertex 
was considered by Vidyasankar~\cite{V78}. 
He showed the necessary and sufficient condition 
in terms of linear inequalities 
that there exists a feasible solution of this problem
(a weaker version was shown by Frank~\cite{F79}).
However, to the best of our knowledge, 
an algorithm for $\mbox{CDGI}(D)$ was not 
presented.

We will summerize our results as follows. 
\begin{enumerate}
\item 
We first show that $\mbox{CDGI}(D)$ can be viewed as some type of 
the connectivity augmentation problem. After this, 
we will prove that 
this connectivity augmentation problem can be solved by using 
an algorithm for the weighted matroid intersection problem in time bounded by a polynomial in 
$\sum_{i=1}^d f(s_i)$ and the size of $D$ 
(this generalizes the result by Frank~\cite{F06}). 
\item
For the case where $D$ is acyclic, we show another characterization for $D$ 
that there exists a feasible solution of $\mbox{CDGI}(D)$. Moreover, we prove that 
in this case $\mbox{CDGI}(D)$ can be solved more efficiently than the general case by 
finding maximum matchings in a series of bipartite graphs instead of using 
an algorithm for the weighted matroid intersection problem. \\
\end{enumerate}

\subsection{Outline}

The rest of this paper is organized as follows. 
Section~\ref{Preliminaries} gives necessary definitions and fundamental results. 
In Section~\ref{Algorithm}, we give an algorithm for the problem 
$\mbox{CDGI}$ by using an algorithm for the weighted matroid intersection 
problem. In Section~\ref{Acyclic Case}, we consider the acyclic case. 

\section{Preliminaries}
\label{Preliminaries}

Let $D=(V,A,S,f)$ be a connected directed graph which may have multiple arcs. 
Let $S=\{s_1,\ldots,s_d\}$. 
Since we can always cover by
$|A|$ $(D,s_i)$-in-trees the arc set of the subgraph
of $D$ induced by $V^i_D$, 
we consider the problem by using at most
$|A|$ $(D,s_i)$-in-trees. That is,
without loss of generality, we assume that $f(s_i) \le |A|$.
For $B \subseteq A$, let $\partial^-(B)$ (resp.~$\partial^+(B)$) be a set of tails (resp.~heads) of arcs in $B$. 
For $e \in A$, we write $\partial^-(e)$ and $\partial^+(e)$ instead of $\partial^-(\{e\})$ and $\partial^+(\{e\})$, 
respectively. 
For $W \subseteq V$,
we define $\delta_D(W)=\{e\in A \colon \partial^-(e) \in W, \partial^+(e) \notin W\}$.
For $v \in V$, we write $\delta_D(v)$ instead of $\delta_D(\{v\})$. 
For two distinct vertices $u, v \in D$, we denote by $\lambda(u,v;D)$ the 
local arc connectivity from $u$ to $v$ in $D$, i.e.,
$\lambda(u,v;D)
= \min\{|\delta_D(W)|\colon u \in W, v \notin W, W \subseteq V\}$. 
We call a subgraph $T$ of $D$ {\it forest} if 
$T$ has no cycle when we ignore the direction of arcs in $T$. 
If a forest $T$ is connected, we call $T$ {\it tree}.  
If every arc of an arc set $B$ is parallel to some arc in $A$, 
we say that $B$ is {\it parallel} to $A$. 
We denote a directed graph obtained by adding an arc set $B$ to $A$ by $D+B$, 
i.e., $D+B=(V,A\cup B,S,f)$. 
For $S'\subseteq S$, let $f(S')=\sum_{s_i \in
S'}f(s_i)$. 
For $v \in V$, we denote by $R_D(v)$ a set of vertices in
$S$ which are reachable from $v$ in $D$. 
For $W \subseteq V$,
let $R_D(W)=\bigcup_{v \in W}R_D(v)$.

For an arc set $B$ which is parallel to $A$, we clearly have
for every $v \in V$ 
\begin{equation} \label{eq3:directed graphs}
R_D(v)=R_{D+B}(v). 
\end{equation}
From (\ref{eq3:directed graphs}), we have
for every $i=1,\ldots,d$
\begin{equation} \label{eq4:directed graphs}
V^i_D=V^i_{D+B}.
\end{equation}

We define $D^{\ast}$ as a directed graph  obtained from $D$ by
adding a new vertex $s^{\ast}$ and connecting $s_i$ to $s^{\ast}$ with
$f(s_i)$ parallel arcs for every $i=1,\ldots,d$ (see Figure~\ref{fig:covering1}).
We denote by $A^{\ast}$ the arc set of $D^{\ast}$.  
From the definition of $D^{\ast}$, 
\begin{equation} \label{eq1:directed graphs}
|A^{\ast}|
=\mbox{$\sum$}_{v\in V}|\delta_{D^{\ast}}(v)|
= |A| + f(S).
\end{equation}
We say that $D$ is {\it $(S,f)$-proper} when 
$|\delta_{D^{\ast}}(v)| \le f(R_D(v))$ holds for every $v \in V$. 

\subsection{Rooted arc-connectivity augmentation by reinforcing arcs}

Given a directed graph  $D=(V,A,S,f)$, 
we call an arc set $B$ with $A\cap B=\emptyset$ which is parallel to $A$ 
a {\it $D^{\ast}$-rooted connector} 
if $\lambda(v,s^{\ast};D^{\ast}+B) \ge f(R_D(v))$ holds for every $v \in V$.
Notice that since a $D^{\ast}$-rooted connector $B$ is parallel to $A$, 
$B$ does not contain an arc which is parallel to an arc entering into $s^{\ast}$ in $D^{\ast}$. 
Then, the problem {\it rooted arc-connectivity augmentation by reinforcing arcs} 
(in short $\mbox{RAA-RA}$) is formally defined as follows.
\begin{center}
\begin{tabularx}{150mm}{rX}
\hline
{\bf Problem$\colon$} & $\mbox{RAA-RA}(D^{\ast})$ \\
\hline 
{\bf Input$\colon$} & $D^{\ast}$ of a directed graph $D$;\\
{\bf Output$\colon$} &
a $D^{\ast}$-rooted connector $B$ whose size is minimum among all 
$D^{\ast}$-rooted connectors. \\
\hline
\end{tabularx}
\end{center}

Notice that the problem $\mbox{RAA-RA}(D^{\ast})$ 
is not equivalent to the local arc-connectivity augmentation problem with 
minimum number of reinforcing arcs from $v \in V$ to $s_i \in R_D(v)$.
For example, we consider $D^{\ast}$ illustrated in Figure~\ref{example_raa-ra}(a) 
of a directed graph  $D=(V,A,S,f)$ 
where $S=\{s_1,s_2\}$, $f(s_1)=2$ and $f(s_2)=2$. The broken lines in Figure~\ref{example_raa-ra}(b)
represent a minimum $D^{\ast}$-rooted connector. For the problem that asks to increase the $v$-$s_i$ local 
arc-connectivity for every $v \in V$ and $s_i \in R_D(v)$ to $f(s_i)$ by adding minimum 
parallel arcs to $A$ (this problem is called the problem 
{\it increasing arc-connectivity by reinforcing arcs} in \cite{J97}, 
in short $\mbox{IARA}(D^{\ast})$), 
an optimal solution is a set of broken lines in Figure~\ref{example_raa-ra}(c). 
While it is known~\cite{J97} that $\mbox{IARA}(D^{\ast})$ 
is $\mathcal{NP}$-hard, 
it is known~\cite{F06} that 
$\mbox{RAA-RA}(D^{\ast})$ in which $S$ consists of a single element can be solved 
in time bounded by a polynomial in $f(S)$ and the size of $D$
by using an algorithm for the weighted matroid intersection.    

\begin{figure}[h]
\begin{minipage}{0.33\hsize}
\begin{center}
\includegraphics[width=3cm]{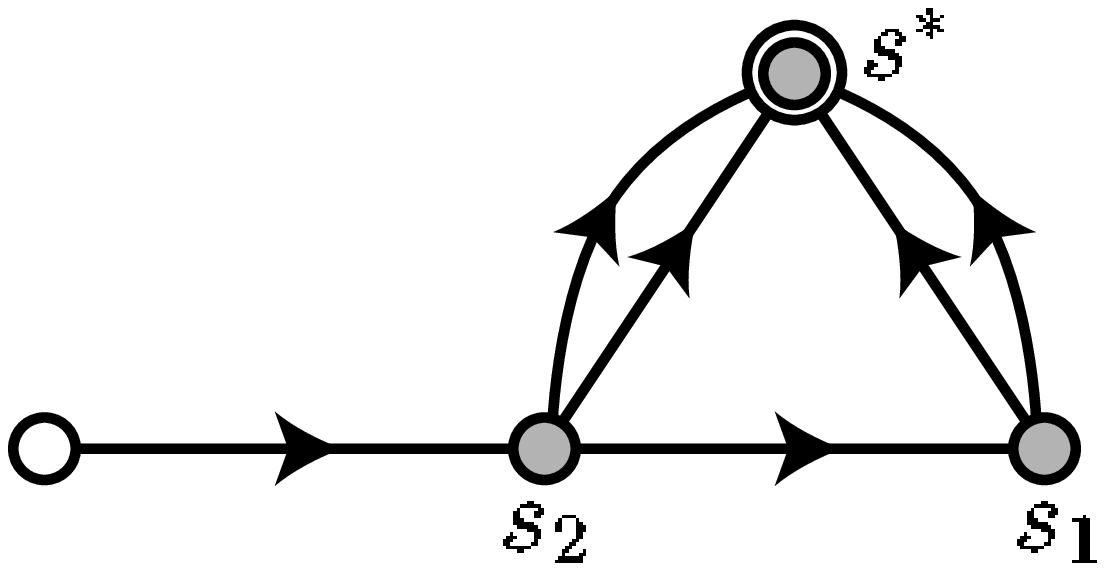}
\par(a)
\end{center}
\end{minipage}
\begin{minipage}{0.32\hsize}
\begin{center}
\includegraphics[width=3cm]{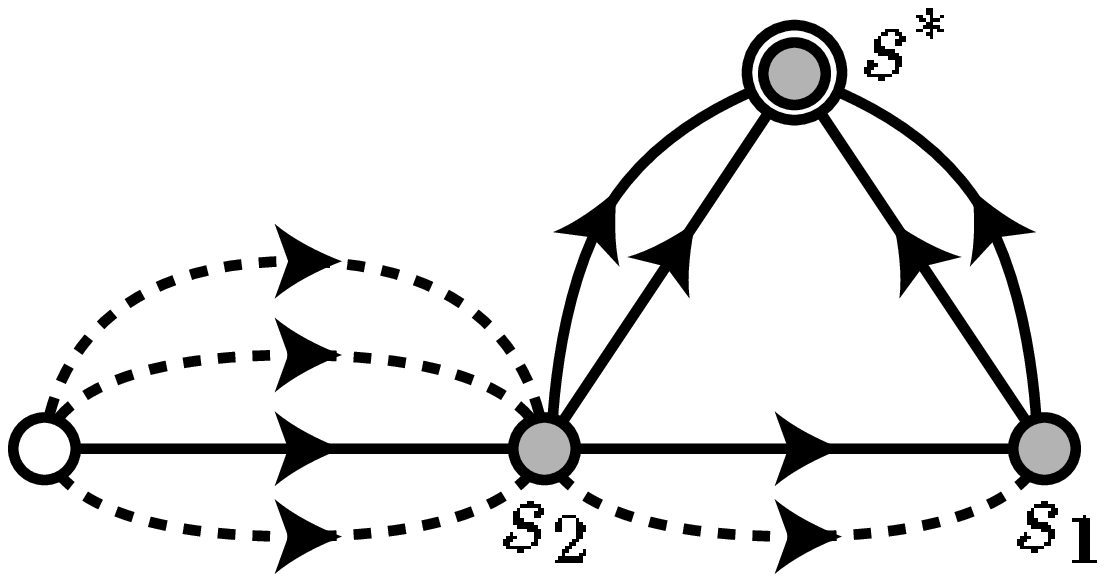}
\par(b)
\end{center}
\end{minipage}
\begin{minipage}{0.33\hsize}
\begin{center}
\includegraphics[width=3cm]{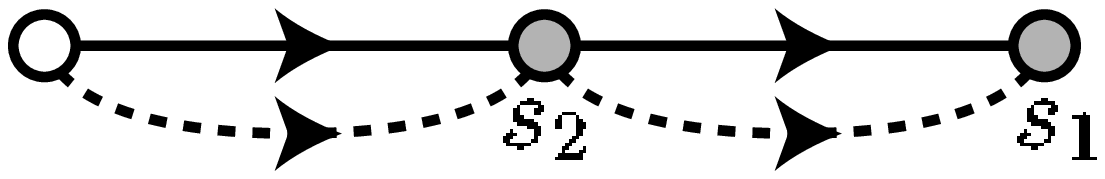}
\par(c)
\end{center}
\end{minipage}
\caption{\small (a) Input.
(b) Optimal solution for $\mbox{RAA-RA}$. (c) Optimal solution for $\mbox{IARA}$.}
\label{example_raa-ra}
\end{figure}

\subsection{Matroids on arc sets of directed graphs}

In this subsection, 
we define two matroids $\bm{M}(D^{\ast})$ and $\bm{U}(D^{\ast})$ 
on $A^{\ast}$ for a directed graph $D=(V,A,S,f)$, 
which will be used in the subsequent discussion.
We denote by $\bm{M}=(E,\mathcal{I})$ a matroid on $E$ whose collection of independent sets is 
$\mathcal{I}$. Introductory treatment of a matroid is given in \cite{O92}. 

For $i=1,\ldots,d$ and $j=1,\ldots,f(s_i)$, 
we define $\bm{M}_{i,j}(D^{\ast})=(A^{\ast},\mathcal{I}_{i,j}(D^{\ast}))$ 
where $I \subseteq A^{\ast}$ belongs to $\mathcal{I}_{i,j}(D^{\ast})$ 
if and only if both of 
a tail and a head of every arc in $I$ are contained in $V^i_D \cup \{s^{\ast}\}$ and 
a directed graph $(V^i_D\cup \{s^{\ast}\}, I)$ is a forest. 
$\bm{M}_{i,j}(D^{\ast})$ is clearly a matroid (i.e.~graphic matroid). 
Moreover, we denote the union of $\bm{M}_{i,j}(D^{\ast})$ 
for $i=1,\ldots,d$ and $j=1,\ldots,f(s_i)$
by $\bm{M}(D^{\ast})=(A^{\ast},\mathcal{I}(D^{\ast}))$ 
in which $I \subseteq A^{\ast}$ belongs to $\mathcal{I}(D^{\ast})$
if and only if $I$ can be partitioned into $\{I_{i,1},\ldots,I_{i,f(s_i)}\colon i=1,\ldots,d\}$ 
such that each $I_{i,j}$ belongs to $\mathcal{I}_{i,j}(D^{\ast})$. 
$\bm{M}(D^{\ast})$ is also a matroid
(see Chapter~12.3 in \cite{O92}. This matroid is also called {\it matroid sum}).
When $I \in \mathcal{I}(D^{\ast})$ can be partitioned into $\{I_{i,1},\ldots,I_{i,f(s_i)}\colon i=1,\ldots,d\}$
such that a directed graph $(V^i_D\cup \{s^{\ast}\},I_{i,j})$ is a tree for 
every $i=1,\ldots,d$ and $j=1,\ldots,f(s_i)$, 
we call $I$ a {\it base of $\bm{M}(D^{\ast})$}. 

Next we define another matroid. We define $\bm{U}(D^{\ast})=(A^{\ast},\mathcal{J}(D^{\ast}))$ 
where $I\subseteq A^{\ast}$ 
belongs to $\mathcal{J}(D^{\ast})$ if and only if $I$ satisfies 
\begin{equation} \label{eq1:matroids}
|\delta_{D^{\ast}}(v)\cap I| \le 
\left\{
\begin{array}{ll}
f(R_D(v)), & \mbox{ if } v \in V, \\
0, & \mbox{ if } v=s^{\ast}.  
\end{array}
\right.
\end{equation}
Since $\bm{U}(D^{\ast})$ is a direct sum of uniform matroids,  
$\bm{U}(D^{\ast})$ is also a matroid (see Exercise~7 of pp.16 and Example~1.2.7 in \cite{O92}). 
We call $I \in \mathcal{J}(D^{\ast})$ a {\it base of $\bm{U}(D)$} 
when (\ref{eq1:matroids}) holds with equality. 

For two matroids $\bm{M}(D^{\ast})$ and $\bm{U}(D^{\ast})$, 
we call an arc set $I \subseteq A^{\ast}$
{\it $D^{\ast}$-intersection} when $I \in \mathcal{I}(D^{\ast})\cap \mathcal{J}(D^{\ast})$. 
If a $D^{\ast}$-intersection $I$ is a base of both $\bm{M}(D^{\ast})$ and $\bm{U}(D^{\ast})$, 
we call $I$ {\it complete}. 

When we are given a weight function $w \colon A^{\ast}\to \mathbb{R}_+$ where $\mathbb{R}_+$ 
denotes the set of non-negative reals, 
we define the weight of $I\subseteq A^{\ast}$ (denoted by $w(I)$) 
by the sum of weights of all arcs $I$. 
The {\it weighted matroid intersection problem} (in short $\mbox{WMI}$) is then defined as follows~\cite{F81}. 
\begin{center}
\begin{tabularx}{150mm}{rX}
\hline
{\bf Problem$\colon$} & $\mbox{WMI}(D^{\ast})$ 
\\
\hline 
{\bf Input$\colon$} &
$D^{\ast}$ of a directed graph  $D$ 
and a weight function $w\colon A^{\ast}\to \mathbb{R}_+$;
\\
{\bf Output$\colon$} &
a complete $D^{\ast}$-intersection $I$ whose weigh is minimum among all complete 
$D^{\ast}$-intersections, if one exists.
\\
\hline
\end{tabularx}
\end{center}

\begin{lemma} \label{lemma:matroid}
We can solve $\mbox{WMI}(D^{\ast})$ in 
$O(M|A^{\ast}|^6)$ time where $M=\sum_{v \in V}f(R_D(v))$. 
\end{lemma}
\begin{proof}
To prove the lemma, we use the following theorem concerning a matroid. 

\begin{theorem}[\cite{K74}] \label{theorem:matroid-knuth}
Given a matroid $\bm{M}=(E,\mathcal{I})$ 
which is a union of $t$ ($\le |E|$) matroids $\bm{M}_1=(E,\mathcal{I}_1),\ldots,\bm{M}_t=(E,\mathcal{I}_t)$, 
we can test if a given set belongs to 
$\mathcal{I}$ in $O(|E|^3\gamma)$ time where $\gamma$ is the time required 
to test if a given set belongs to $\mathcal{I}_1,\ldots,\mathcal{I}_t$. 
\end{theorem}

\begin{theorem}[\cite{F81}] \label{theorem:matroid-frank}
Given two matroids $\bm{M}_1=(E,\mathcal{I}_1)$ and $\bm{M}_2=(E,\mathcal{I}_2)$ with a weight function 
$w\colon E \to \mathbb{R}_+$ and a non-negative integer $k \in \mathbb{Z}_+$, 
we can find $I\in \mathcal{I}_1\cap \mathcal{I}_2$ with $|I|=k$ whose weight is minimum among all 
$I'\in \mathcal{I}_1 \cap \mathcal{I}_2$ with $|I'|=k$ in $O(k|E|^3 + k|E|^2\gamma)$ time if one 
exists where $\gamma$ is the time required to test if 
a given set belongs to both $\mathcal{I}_1$ and $\mathcal{I}_2$. 
\end{theorem}

We consider the time required to test if a given set 
belongs to both $\mathcal{I}(D^{\ast})$ and $\mathcal{J}(D^{\ast})$. 
Since it is not difficult to see that 
we can test is a given set belongs to each $\mathcal{I}_{i,j}(D^{\ast})$
in $O(|A^{\ast}|)$ time, 
we can test if a given set belongs to $\mathcal{I}(D^{\ast})$ in
$O(|A^{\ast}|^4)$ time from Theorem~\ref{theorem:matroid-knuth}. 
For $\mathcal{J}(D^{\ast})$, the time complexity is clearly $O(|A^{\ast}|)$ time. 
The size of every complete $D^{\ast}$-intersection is equal to 
$M$ from (\ref{eq1:matroids}). 
From this discussion, the total time required for solving $\mbox{WMI}(D^{\ast})$ is $O(M|A^{\ast}|^6)$ from 
Theorem~\ref{theorem:matroid-frank}. 
\end{proof}

\subsection{Results from \cite{KKT08}}

In this section, we introduce results concerning packing of in-trees given by 
Kamiyama et al.~\cite{KKT08} which plays a crucial role in this paper.  

\begin{theorem}[\cite{KKT08}] \label{theorem:KKT08}
Given a directed graph  $D=(V,A,S,f)$, the following three statements are equivalent$\colon$
\begin{enumerate}
\item For every $v \in V$, $\lambda(v,s^{\ast};D^{\ast})\ge f(R_D(v))$ holds.
\item There exists a $D$-canonical set of arc-disjoint in-trees. 
\item There exists a complete $D^{\ast}$-intersection. 
\end{enumerate}
\end{theorem}

Although the following theorem is not explicitly proved in \cite{KKT08},
we can easily obtain it from the proof of Theorem~\ref{theorem:KKT08} in \cite{KKT08}. 

\begin{theorem}[\cite{KKT08}] \label{theorem2:KKT08}
Given a directed graph  $D=(V,A,S,f)$ which satisfies the condition of 
Theorem~\ref{theorem:KKT08}, we can find a $D$-canonical set of arc-disjoint
in-trees in $O(M^2|A|^2)$ time where $M=\sum_{v \in V}f(R_D(v))$.
\end{theorem}

From Theorem~\ref{theorem:KKT08}, we obtain the following corollary. 

\begin{corollary} \label{corollary:KKT08}
Given a directed graph  $D=(V,A,S,f)$ and an arc set $B$ with $A\cap B=\emptyset$ 
which is parallel to 
$A$, the following three statements are equivalent$\colon$
\begin{enumerate}
\item $B$ is a $D^{\ast}$-rooted connector.
\item There exists a $(D+B)$-canonical set of arc-disjoint in-trees. 
\item There exists a complete $(D+B)^{\ast}$-intersection. 
\end{enumerate}
\end{corollary}
\begin{proof}
The equivalence of the statements~2 and 3 follows from Theorem~\ref{theorem:KKT08}.\\
{\bf 1$\to$2$\colon$}
Since $B$ is parallel to $A$, we clearly have 
\begin{equation} \label{eq2:directed graphs}
(D+B)^{\ast}=D^{\ast}+B. 
\end{equation}
Since $B$ is a $D^{\ast}$-rooted connector,
we have for every $v \in V$
\begin{equation*}
\lambda(v,s^{\ast};(D+B)^{\ast})
\underbrace{=\lambda(v,s^{\ast};D^{\ast}+B)}_{\mbox{\small by (\ref{eq2:directed graphs})}}
\ge f(R_D(v))
\underbrace{=f(R_{D+B}(v))}_{\mbox{by \small (\ref{eq3:directed graphs})}}.
\end{equation*}
From this inequality and Theorem~\ref{theorem:KKT08}, 
this part follows. \\
{\bf 2$\to$1$\colon$}
Since there exists a $(D+B)$-canonical set of  
arc-disjoint in-trees,
we have for every $v\in V$
\begin{equation*}
\lambda(v,s^{\ast};D^{\ast}+B)
\underbrace{=\lambda(v,s^{\ast};(D+B)^{\ast})}_{\mbox{\small by (\ref{eq2:directed graphs})}}
\underbrace{\ge f(R_{D+B}(v))}_{\mbox{\small by Theorem~\ref{theorem:KKT08}}}
\underbrace{=f(R_D(v))}_{\mbox{\small by (\ref{eq3:directed graphs})}}. 
\end{equation*}
This proves that $B$ is a $D^{\ast}$-rooted connector.
\end{proof}

\section{An Algorithm for Covering by In-trees}
\label{Algorithm}

Given a directed graph  $D=(V,A,S,f)$, we present in this section an algorithm for $\mbox{CDGI}(D)$. 
The time complexity of the proposed algorithm is bounded by a polynomial in $f(S)$ and the size of $D$. 
We first prove that $\mbox{CDGI}(D)$ can be reduced to $\mbox{RAA-RA}(D^{\ast})$. 
After this, we show that $\mbox{RAA-RA}(D^{\ast})$ can be solved by using 
an algorithm for the weighted matroid intersection problem. 

\subsection{Reduction from $\mbox{CDGI}$ to $\mbox{RAA-RA}$}
\label{Reduction from CDGI to RAA-RA}

If $D=(V,A,S,f)$ is not $(S,f)$-proper, i.e., 
$|\delta_{D^{\ast}}(v)|>f(R_D(v))$ for some $v \in V$,  
there exists no feasible solution of $\mbox{CDGI}(D)$ 
since there can not be a $D$-canonical set of in-trees 
that covers $\delta_{D^{\ast}}(v)$ 
from the definition of a $D$-canonical set of in-trees.  
Thus, we assume in the subsequent discussion that $D$ is $(S,f)$-proper. 

\begin{proposition} \label{proposition1:raa-ra}
Given an $(S,f)$-proper directed graph $D=(V,A,S,f)$, 
the size of a $D^{\ast}$-rooted connector is at least 
$\mbox{$\sum$}_{v \in V}f(R_D(v)) - (|A|+f(S))$.
\end{proposition}
\begin{proof}
Let $B$ be a $D^{\ast}$-rooted connector. 
For every $v\in V$, $|\delta_{D^{\ast}+B}(v)|\ge f(R_D(v))$ holds 
from the definition of a $D^{\ast}$-rooted connector.
Thus, the number of arcs of $D^{\ast}+B$ is at least $\sum_{v \in V}f(R_D(v))$. 
Since the number of arcs of $D^{\ast}$ is equal to 
$|A|+f(S)$ from (\ref{eq1:directed graphs}), the proposition holds. 
\end{proof}
For an $(S,f)$-proper directed graph $D=(V,A,S,f)$, we define ${\sf opt}_D$ by 
\begin{equation} \label{eq1:rooted}
{\sf opt}_D=\mbox{$\sum$}_{v \in V}f(R_D(v)) - (|A|+f(S)).
\end{equation}
From Proposition~\ref{proposition1:raa-ra}, the size of a $D^{\ast}$-rooted 
connector is at least ${\sf opt}_D$. 

\begin{lemma} \label{lemma2:raa-ra}
Given an $(S,f)$-proper
directed graph  $D=(V,A,S,f)$, there exists a feasible solution of $\mbox{CDGI}(D)$
if and only if there exists a $D^{\ast}$-rooted connector whose size is equal to ${\sf opt}_D$. 
\end{lemma}
\begin{proof}
{\bf Only if-part$\colon$}
Suppose there exists a feasible solution of $\mbox{CDGI}(D)$, i.e., 
there exists a $D$-canonical set $\mathcal{T}$ of in-trees which covers $A$. 
For each $i=1,\ldots,d$, we denote $f(s_i)$ $(D,s_i)$-in-trees of $\mathcal{T}$ 
by $T_{i,1},\ldots,T_{i,f(s_i)}$.
For each $e \in A$, let $P_e=\{(i,j)\colon e \mbox{ is contained in }T_{i,j}\}$.
Since $\mathcal{T}$ covers $A$, 
each $e \in A$ is contained in at least one in-tree of $\mathcal{T}$. 
Thus, $|P_e|\ge 1$ holds for every $e \in A$. 
We define an arc set $B$ by 
$B = \bigcup_{e \in A}\{|P_e| - 1\mbox{ copies of }e\}$. 
We will prove that $B$ is a $D^{\ast}$-rooted 
connector whose size is equal to ${\sf opt}_D$.  

We first prove $|B|={\sf opt}_D$. For this, 
we show that for every $v \in V$ 
\begin{equation} \label{claim1:lemma2:raa-ra}
\mbox{$\sum$}_{e \in \delta_D(v)}(|P_e|-1)=f(R_D(v))-|\delta_{D^{\ast}}(v)|.
\end{equation}
Let us first consider $v \notin S$. 
For $s_i \in R_D(v)$, $T_{i,j}$ contains $v$
since $T_{i,j}$ spans $V^i_D$ and $s_i$ is reachable from $v$. 
Hence, 
since $T_{i,j}$ is an in-tree and $v$ is not a root of $T_{i,j}$ from $v \notin S$, 
$T_{i,j}$ contains exactly one arc $e \in \delta_{D}(v)$, i.e., 
$(i,j)$ is contained in $P_e$ for exactly one arc $e \in \delta_D(v)$. 
Thus, 
$\mbox{$\sum$}_{e \in \delta_D(v)}|P_e|=
\sum_{s_i \in R_D(v)}f(s_i)
=f(R_D(v))$.
From this equation and since $|\delta_D(v)|=|\delta_{D^{\ast}}(v)|$ follows from $v \notin S$, 
(\ref{claim1:lemma2:raa-ra}) holds. 
In the case of $v \in S$, for $s_i \in R_D(v)\setminus \{v\}$, 
$(i,j)$ is contained in $P_e$ for exactly one arc
$e \in \delta_D(v)$ as in the case of $v \notin S$. Thus, 
$\mbox{$\sum$}_{e \in \delta_D(v)}|P_e|=f(R_D(v))-f(v)$. 
From this equation and $|\delta_{D^{\ast}}(v)|=|\delta_D(v)|+f(v)$, 
\begin{align*} 
\mbox{$\sum$}_{e \in \delta_D(v)}(|P_e|-1)
= f(R_D(v))-f(v)-|\delta_D(v)| 
= f(R_D(v))-|\delta_{D^{\ast}}(v)|. 
\end{align*}
This completes the proof of (\ref{claim1:lemma2:raa-ra}).
Since $B$ contains $|P_e|-1$ copies of $e\in A$, 
\begin{align*} 
|B| &= \mbox{$\sum$}_{v \in V}\mbox{$\sum$}_{e \in \delta_D(v)}(|P_e|-1) \nonumber \\
&=\mbox{$\sum$}_{v \in V}(f(R_D(v))-|\delta_{D^{\ast}}(v)|) 
\ \ \ (\mbox{from (\ref{claim1:lemma2:raa-ra}}))\nonumber \\
&= {\sf opt}_D \ \ \ (\mbox{from (\ref{eq1:directed graphs}) and (\ref{eq1:rooted})}). 
\end{align*}

What remains is to prove that $B$ is a $D^{\ast}$-rooted connector. 
From Corollary~\ref{corollary:KKT08}, it is sufficient to  
prove that there exists a $(D+B)$-canonical set of arc-disjoint in-trees.
For this, we will construct from $\mathcal{T}$ a set $\mathcal{T}'$ of arc-disjoint in-trees 
which consists of $T'_{i,1},\ldots,T'_{i,f(s_i)}$ for $i=1,\ldots,d$,
and we prove that $\mathcal{T}'$ is a $(D+B)$-canonical set of in-trees.
Each $T'_{i,j}$ is constructed from $T_{i,j}$ as follows. 
When $e \in A$ is contained in more than one in-tree of $\mathcal{T}$, 
in order to construct $\mathcal{T}'$ from $\mathcal{T}$, 
we need to replace $e$ of $T_{i,j}$ by an arc in $B$ which is parallel to $e$
for every $(i,j)\in P_e$ except one in-tree.
For $(i_{\min},j_{\min})\in P_e$ which is lexicographically smallest in $P_e$, 
we allow $T'_{i_{\min},j_{\min}}$ to use $e$, while for $(i,j)\in P_e\setminus (i_{\min},j_{\min})$, 
we replace $e$ of $T_{i,j}$ by an arc in $B$ which is parallel to $e$ 
so that for distinct $(i,j), (i',j') \in P_e\setminus (i_{\min},j_{\min})$, 
the resulting $T'_{i,j}$ and $T'_{i',j'}$ contain distinct arcs which are parallel to $e$, 
respectively (see Figure~\ref{example_replace}).

\begin{figure}[h]
\begin{center}
\includegraphics[width=1.5cm]{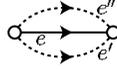}
\end{center}
\caption{\small Illustration of the replacing operation. Let $e$ be an arc in $A$, and 
let $e',e''$ be arcs in $B$. Assume that $P_e=\{(1,1),(1,2),(2,1)\}$. In this case, $T_{1,1}$, 
$T_{1,2}$ and $T_{2,1}$ contain $e$. Then, $T'_{1,1}$ contains $e$, $T'_{1,2}$ contains $e'$, 
and $T'_{2,1}$ contains $e''$.}
\label{example_replace}
\end{figure}
 
We will do this operation for every $e \in A$.
Let $\mathcal{T}'$ be the set of in-trees obtained by performing the above operation for every $e \in A$. 
Here we show that $\mathcal{T}'$ is a $(D+B)$-canonical set of arc-disjoint in-trees. 
Since $T'_{i,j}$ and $T'_{i',j'}$ are arc-disjoint for $(i,j)\neq (i',j')$ 
from the way of constructing $\mathcal{T}'$, 
it is sufficient to prove that $T'_{i,j}$ is a $(D+B,s_i)$-in-tree. 
Since $T'_{i,j}$ is constructed by replacing arcs of $T_{i,j}$ by 
the corresponding parallel arc in $B$ and 
$T_{i,j}$ is an in-tree rooted at $s_i$,  
$T'_{i,j}$ is also an in-tree rooted at $s_i$. Since $T_{i,j}$ spans $V^i_D$ and from (\ref{eq4:directed graphs}), 
$T'_{i,j}$ spans $V^i_{D+B}$. Hence, $T'_{i,j}$ 
is a $(D+B,s_i)$-in-tree. This completes the proof. \\
{\bf If-part$\colon$}
Let $B$ be a $D^{\ast}$-rooted connector with $|B|={\sf opt}_D$. 
From Corollary~\ref{corollary:KKT08}, there exists a $(D+B)$-canonical set $\mathcal{T}'$ of arc-disjoint 
in-trees. For each $i=1,\ldots,d$, we denote $f(s_i)$ $(D+B,s_i)$-in-trees of $\mathcal{T}'$ 
by $T'_{i,1},\ldots,T'_{i,f(s_i)}$.
We will prove that we can construct from $\mathcal{T}'$ 
a $D$-canonical set of in-trees covering $A$. 
We first construct from $\mathcal{T}'$ a set $\mathcal{T}$ of in-trees which consists of $T_{i,j}$ for $i=1,\ldots,d$
and $j=1,\ldots,f(s_i)$ by the following procedure {\sf Replace}.
\begin{center}
\fbox{
\begin{minipage}{150mm}
{\bf Procedure {\sf Replace}$\colon$}
For each $i=1,\ldots,d$ and $j=1,\ldots,f(s_i)$, 
set $T_{i,j}$ to be a directed graph obtained from $T'_{i,j}$
by replacing every arc $e \in B$ which is contained in $T'_{i,j}$
by an arc in $A$ which is parallel to $e$. 
\end{minipage}
}
\end{center}

From now on, we prove that $\mathcal{T}$ is a $D$-canonical set of 
in-trees which covers $A$. 
It is not difficult to prove that $\mathcal{T}$ is a $D$-canonical set of in-trees 
from the definition of the procedure {\sf Replace} in the same manner as the last part 
of the proof of the ``only if-part''. 
Thus, it is sufficient to prove that $\mathcal{T}$ covers $A$. 
For this, we first show that 
$\mathcal{T}'$ covers $A\cup B$. 
From $A \cap B=\emptyset$, $|B|={\sf opt}_D$ and (\ref{eq1:rooted}), 
\begin{equation} \label{eq4:lemma1:raa-ra}
|A\cup B| = |A| + {\sf opt}_D = \mbox{$\sum$}_{v \in V}f(R_D(v)) - f(S). 
\end{equation}
Recall that each $v \in V$ is contained in $f(R_{D+B}(v))$ in-trees of $\mathcal{T}'$
from the definition of a $(D+B)$-canonical set of in-trees. 
Thus, since in-trees of $\mathcal{T}'$ are arc-disjoint, 
it holds for each $v \in V$ that  
the number of arcs in $\delta_{D+B}(v)$
which are contained in in-trees of $\mathcal{T}'$ is equal to 
\begin{equation}
\left\{
\begin{array}{ll}
f(R_{D+B}(v)), & \mbox{ if } v \in V\setminus S,\\
f(R_{D+B}(v)) - f(v), & \mbox{ if } v \in S.  
\end{array}
\right.
\end{equation}
Hence, the number of arcs in $A\cup B$ contained in in-trees of 
$\mathcal{T}'$ is equal to
\begin{align} \label{eq6:lemma1:raa-ra}
&\mbox{$\sum$}_{v \in V\setminus S}f(R_{D+B}(v))+ \mbox{$\sum$}_{v \in S}(f(R_{D+B}(v))-f(v)) \nonumber \\
&=\mbox{$\sum$}_{v \in V}f(R_{D+B}(v)) - f(S) 
=\mbox{$\sum$}_{v \in V}f(R_D(v))-f(S) \ \ \ (\mbox{from (\ref{eq3:directed graphs})}). 
\end{align}
Since any arc of $\mathcal{T}'$ is in $A\cup B$ and the number of arcs in $A\cup B$
is equal to that of $\mathcal{T}'$
from (\ref{eq4:lemma1:raa-ra}) and (\ref{eq6:lemma1:raa-ra}),
$\mathcal{T}'$ contains all arcs in $A$. 
Thus, $\mathcal{T}$ covers $A$ from the definition of the procedure {\sf Replace}. 
\end{proof}

As seen in the proof of the ``if-part'' of Lemma~\ref{lemma2:raa-ra}, 
if we can find a $D^{\ast}$-rooted connector $B$ with $|B|={\sf opt}_D$, 
we can compute a $D$-canonical set of in-trees which covers $A$ by using the procedure 
{\sf Replace} from a $(D+B)$-canonical set of arc-disjoint in-trees. 
Furthermore, we can construct a $(D+B)$-canonical set of arc-disjoint in-trees by using 
the algorithm of Theorem~\ref{theorem2:KKT08}. 
Since the optimal value of $\mbox{RAA-RA}(D^{\ast})$ is at least ${\sf opt}_D$ from Proposition~\ref{proposition1:raa-ra}, 
we can test if there exists a $D^{\ast}$-rooted connector whose size is equal to ${\sf opt}_D$
by solving $\mbox{RAA-RA}(D^{\ast})$. 
Assuming that we can solve $\mbox{RAA-RA}(D^{\ast})$, 
our algorithm for finding a $D$-canonical set of in-trees which covers $A$
called Algorithm {\sf CR} can be illustrated as Algorithm~\ref{Algorithm:Algorithm1} below. 

\begin{algorithm}[h]
\begin{algorithmic}[1]
\INPUT a directed graph $D=(V,A,S,f)$
\OUTPUT a $D$-canonical set of in-trees covering $A$, if one exists
\IF{$D$ is not $(S,f)$-proper}
\STATE Halt (there exists no $D$-canonical set of in-trees covering $A$)
\ENDIF
\STATE Find an optimal solution $B$ of $\mbox{RAA-RA}(D^{\ast})$
\IF{$|B| > {\sf opt}_D$}
\STATE Halt (there exists no $D$-canonical set of in-trees covering $A$)
\ELSE
\STATE Construct a $(D+B)$-canonical set $\mathcal{T}'$ of arc-disjoint in-trees
\STATE Construct a set $\mathcal{T}$ of in-trees from $\mathcal{T}'$ by using the procedure {\sf Replace}
\RETURN $\mathcal{T}$ 
\ENDIF
\end{algorithmic}
\caption{Algorithm {\sf CR}}
\label{Algorithm:Algorithm1}
\end{algorithm}

\begin{lemma} \label{lemma1:computation}
Given a directed graph $D=(V,A,f,S)$, Algorithm ${\sf CR}$ correctly 
finds a $D$-canonical set of in-trees which covers $A$ in $O(\gamma_1 + |V||A|+ M^4)$ time if one exists where 
$\gamma_1$ is the time required to solve $\mbox{RAA-RA}(D^{\ast})$ and 
$M=\sum_{v \in V}f(R(v))$. 
\end{lemma}
\begin{proof}
The correctness of the algorithm follows from Lemma~\ref{lemma2:raa-ra}. Thus, 
we consider the time complexity. 
In Step~1, we have to compute $R_D(v)$ for every $v \in V$. 
This can be done in $O(|V||A|)$ time by applying depth-first search from every $s_i \in S$. 
After this, the time required to test whether $|\delta_{D^{\ast}}(v)| \le f(R_D(v))$ for all $v\in V$ 
is $O(|A|)$. Thus, the time required for Step~1 is $O(|V||A|)$. 
Since the number of arcs of $D+B$ is at most $M$ for a $D^{\ast}$-rooted connector $B$ with
$|B|={\sf opt}_D$ from (\ref{eq1:rooted}), the time required for Step~8 is $O(M^4)$
from Theorem~\ref{theorem2:KKT08}. 
Moreover, since the number of arcs of $D+B$ is at most $M$, the time required for Step~9
is $O(M)$ from the definition of Procedure {\sf Replace}.
Hence, since the time required for Step~4 is $\gamma_1$, the lemma follows. 
\end{proof}

\subsection{Reduction from $\mbox{RAA-RA}$ to $\mbox{WMI}$}

From the algorithm {\sf CR} in Section~\ref{Reduction from CDGI to RAA-RA},
in order 
to present an algorithm for $\mbox{CDGI}(D)$, what remains is to 
show how we solve $\mbox{RAA-RA}(D^{\ast})$. 
In this section, we will prove that we can test whether there exists 
a $D^{\ast}$-rooted connector whose size is equal to ${\sf opt}_D$ 
(i.e., Steps~4 and 5 in the algorithm {\sf CR}) 
by reducing it to the problem $\mbox{WMI}$. 
Our proof is based on the algorithm of \cite{F06} for 
$\mbox{RAA-RA}(D^{\ast})$ in which $S$ consists of a single vertex.
We extend the idea of \cite{F06} to the case of 
$|S|>1$ by using Theorem~\ref{theorem:KKT08}.  
We define a directed graph $D_+$ obtained from $D$ by adding ${\sf opt}_D$ parallel arcs to every 
$e \in A$. 
Then, we will compute a $D^{\ast}$-rooted connector whose size is equal to ${\sf opt}_D$ 
by using an algorithm for $\mbox{WMI}(D^{\ast}_+)$
as described below. 
Since the number of arcs in a $D^{\ast}$-rooted connector whose size is equal to ${\sf opt}_D$ 
which are parallel to one arc in $A$ 
is at most ${\sf opt}_D$, it is enough to add ${\sf opt}_D$ parallel arcs to 
each arc of $A$ in $D_+$ in order to find a $D^{\ast}$-rooted connector whose size 
is equal to ${\sf opt}_D$. 

We denote by $A_+$ and $A^{\ast}_+$ the arc sets of $D_+$ and $D^{\ast}_+$, respectively. 
If $I\subseteq A_+^{\ast}$ is a complete 
$D_+^{\ast}$-intersection, 
since $I$ is a base of $\bm{U}(D_+^{\ast})$ and 
from (\ref{eq1:matroids}) and (\ref{eq3:directed graphs}),  
\begin{equation} \label{eq1:proposition1:wmi}
|I| 
=\mbox{$\sum$}_{v \in V}f(R_{D_+}(v))
=\mbox{$\sum$}_{v \in V}f(R_D(v)). 
\end{equation}
We define a weight function $w \colon A_+^{\ast} \to \mathbb{R}_+$ by  
\begin{equation} \label{eq1:algorithm for raa-ra} 
w(e)= 
\left\{
\begin{array}{ll}
0, & \mbox{ if } e \in A^{\ast}, \\
1, & \mbox{ otherwise}. 
\end{array}
\right.
\end{equation}
The following lemma shows the relation between $\mbox{RAA-RA}(D^{\ast})$ and $\mbox{WMI}(D_+^{\ast})$. 
\begin{lemma} \label{lemma2:wmi}
Given an $(S,f)$-proper directed graph  $D=(V,A,S,f)$, 
there exists a $D^{\ast}$-rooted connector whose size is equal to ${\sf opt}_D$
if and only if 
there exists a complete $D_+^{\ast}$-intersection whose weight is equal to ${\sf opt}_D$.
\end{lemma}
To prove Lemma~\ref{lemma2:wmi}, we need to show the following two lemmas.

\begin{lemma} \label{lemma1:wmi}
Given a directed graph  $D=(V,A,S,f)$ and an arc set $B$ which is parallel to $A$,
\begin{enumerate}
\item if there is a complete $D^{\ast}$-intersection $I$, $I$ is also a complete $(D+B)^{\ast}$-intersection, and
\item if there is a complete $(D+B)^{\ast}$-intersection $I$ such that $I\subseteq A^{\ast}$, $I$ 
is also a complete $D^{\ast}$-intersection.  
\end{enumerate}
\end{lemma}
\begin{proof}
{\bf 1$\colon$}
We first prove that $I$ is a base of $\bm{M}((D+B)^{\ast})$. 
Since $I$ is a base of $\bm{M}(D^{\ast})$, $I$ can be partitioned into $\{I_{i,1},\ldots,I_{i,f(s_i)}\colon i=1,\ldots,d\}$
such that a directed graph $(V_D^i\cup \{s^{\ast}\},I_{i,j})$ is a tree for every 
$i=1,\ldots,d$ and $j=1,\ldots,f(s_i)$.
Thus, since each $(V^i_{D+B}\cup \{s^{\ast}\},I_{i,j})$ is a tree from (\ref{eq4:directed graphs}),
$I$ is a base of $\bm{M}((D+B)^{\ast})$. 

Next we prove that $I$ is a base of $\bm{U}((D+B)^{\ast})$. 
Since $I$ is a base of $\bm{U}(D^{\ast})$,
$|\delta_{D^{\ast}}(v)\cap I|$ is equal to 
\begin{equation*}
\left\{
\begin{array}{ll}
f(R_D(v)), & \mbox{ if } v \in V,\\
0, & \mbox{ if } v=s^{\ast}.
\end{array}
\right.
\end{equation*}
Furthermore, since $I\cap B =\emptyset$ follows from $I\subseteq A^{\ast}$, 
$|\delta_{D^{\ast}}(v)\cap I|$ is equal to $|\delta_{(D+B)^{\ast}}(v)\cap I|$ for every $v \in V$. 
Thus, for each $v\in V$, $|\delta_{(D+B)^{\ast}}(v)\cap I|$ is equal to 
\begin{equation}
\left\{
\begin{array}{ll}
f(R_D(v)) = f(R_{D+B}(v)), & \mbox{ if } v \in V,\\
0, & \mbox{ if } v = s^{\ast}.
\end{array}
\right.
\end{equation}
This proves that $I$ is a base of $\bm{U}((D+B)^{\ast})$.  
\\
{\bf 2$\colon$}
This part can be proved in the same manner as in the proof of the part~1.  
\hfill\usebox{\ProofSym}
\end{proof}

\begin{lemma} \label{proposition1:wmi}
Given $D_+^{\ast}$ of an $(S,f)$-proper directed graph 
$D=(V,A,S,f)$ and a weight function $w\colon A_+^{\ast} \to \mathbb{R}_+$ 
defined by (\ref{eq1:algorithm for raa-ra}),
if there exists a complete $D_+^{\ast}$-intersection $I\subseteq A^{\ast}_+$, $w(I) \ge {\sf opt}_D$. 
Moreover, $w(I)={\sf opt}_D$ if and only if $A^{\ast}\subseteq I$. 
\end{lemma}
\begin{proof}
From (\ref{eq1:algorithm for raa-ra}), we have 
$w(I)=|I| - |I\cap A^{\ast}|$.  
Furthermore, 
\begin{equation*}
|I|-|I\cap A^{\ast}| \ge |I| -|A^{\ast}|
\underbrace{ 
=\mbox{$\sum$}_{v \in V}f(R_D(v))-(|A|+f(S))
}_{\mbox{\small from (\ref{eq1:directed graphs}) and (\ref{eq1:proposition1:wmi})}}.
\end{equation*} 
Thus, $w(I)\ge {\sf opt}_D$ follows from (\ref{eq1:rooted}). 
From the above equation, $w(I)={\sf opt}_D$ if and only if $|I\cap A^{\ast}|=|A^{\ast}|$. 
This proves the rest of the lemma. 
\hfill\usebox{\ProofSym}
\end{proof}

\begin{proof2}{Lemma~\ref{lemma2:wmi}}
{\bf Only if-part$\colon$}
Assume that there exists a $D^{\ast}$-rooted connector whose size is equal to ${\sf opt}_D$.
Since $D_+$ has ${\sf opt}_D$ parallel arcs to every $e \in A$, there exists 
a $D^{\ast}$-rooted connector $B\subseteq A_+\setminus A$ with $|B|={\sf opt}_D$. 
Let us fix a $D^{\ast}$-rooted connector $B\subseteq A_+\setminus A$ with $|B|={\sf opt}_D$. 
From (i) of Lemma~\ref{lemma1:wmi}, in order to prove the ``only if-part'', it is sufficient to prove that there exists 
a complete $(D+B)^{\ast}$-intersection $I$ with $w(I)={\sf opt}_D$. 
Since there exists a complete $(D+B)^{\ast}$-intersection $I$ from Corollary~\ref{corollary:KKT08}, 
we will prove that $w(I)={\sf opt}_D$. 
Since the arc set of $(D+B)^{\ast}$ is equal to $A^{\ast}\cup B$ and 
$I$ is a $(D+B)^{\ast}$-intersection, $I \subseteq A^{\ast} \cup B$ holds. 
Thus, since $w(A^{\ast} \cup B)=|B|={\sf opt}_D$ follows from (\ref{eq1:algorithm for raa-ra}), 
$w(I)\le w(A^{\ast} \cup B)={\sf opt}_D$ holds. 
Hence, $w(I)={\sf opt}_D$ follows from Lemma~\ref{proposition1:wmi}. This completes the proof. 
\\
{\bf If-part$\colon$}
Assume that there exists a complete $D_+^{\ast}$-intersection 
$I$ with $w(I)={\sf opt}_D$. 
Let $B$ be $I \setminus A^{\ast}$, and we will 
prove that $B$ is a $D^{\ast}$-rooted connector with 
$|B| = {\sf opt}_D$. 
We first prove $B$ is a $D^{\ast}$-rooted connector by using (ii) of 
Lemma~\ref{lemma1:wmi} and Corollary~\ref{corollary:KKT08}. 
We set $B$ and $D$ in Lemma~\ref{lemma1:wmi} to be 
$A_+\setminus (A\cup B)$ and $D+B$, respectively.
Notice that $(D+B)+(A_+\setminus (A\cup B))=D_+$ follows from $B \subseteq A_+$
and $A_+ \setminus (A\cup B)$ is parallel to $A\cup B$.
From $B=I\setminus A^{\ast}$, we have $I \subseteq A^{\ast}\cup B$. 
Thus, $I$ is a complete $(D+B)^{\ast}$-intersection
since $I$ is a complete $D_+^{\ast}$-intersection and from (ii) 
of Lemma~\ref{lemma1:wmi}.  
Hence, from Corollary~\ref{corollary:KKT08}, 
$B$ is a $D^{\ast}$-rooted connector. 

What remains is to prove that $|B|={\sf opt}_D$. 
From Lemma~\ref{proposition1:wmi} and $w(I)={\sf opt}_D$,
$A^{\ast}\subseteq I$ holds. 
Thus, from $B=I\setminus A^{\ast}$ and (\ref{eq1:proposition1:wmi}), 
\begin{equation*}
|B|=|I\setminus A^{\ast}|= |I| - |A^{\ast}| = \mbox{$\sum$}_{v \in V}f(R_D(v))- (|A|+f(S)).
\end{equation*} 
This equation and (\ref{eq1:rooted}) complete the proof. 
\end{proof2}

As seen in the proof of the ``if-part'' of Lemma~\ref{lemma2:wmi}, 
if we can find a complete $D_+^{\ast}$-intersection $I$ 
with $w(I) = {\sf opt}_D$, we can find 
a $D^{\ast}$-rooted connector $B$ with $|B|={\sf opt}_D$ 
by setting $B=I\setminus A^{\ast}$. 
Furthermore, we can obtain a complete $D_+^{\ast}$-intersection 
whose weight is equal to ${\sf opt}_D$ if one exists by using the algorithm for $\mbox{WMI}(D_+^{\ast})$
since the optimal value of $\mbox{WMI}(D_+^{\ast})$ is at least ${\sf opt}_D$ 
from Lemma~\ref{proposition1:wmi}. The formal description of the algorithm 
called Algorithm {\sf RW}
for finding a $D^{\ast}$-rooted connector whose size is equal to ${\sf opt}_D$
is illustrated in Algorithm~\ref{Algorithm:Algorithm2}.

\begin{algorithm}[h]
\begin{algorithmic}[1]
\INPUT $D^{\ast}$ of an $(S,f)$-proper directed graph  $D=(V,A,S,f)$
\OUTPUT a $D^{\ast}$-rooted connector whose size is equal to ${\sf opt}_D$, if one exits
\STATE Find an optimal solution $I$ for $\mbox{WMI}(D_+^{\ast})$ 
with a weight function $w$ defined by (\ref{eq1:algorithm for raa-ra})\\
\IF{there exists no solution of $\mbox{WMI}(D_+^{\ast})$ or $w(I)>{\sf opt}_D$}
\STATE Halt (There exists no $D^{\ast}$-rooted connector whose size is equal to ${\sf opt}_D$)
\ENDIF 
\RETURN $I\setminus A^{\ast}$
\end{algorithmic}
\caption{Algorithm {\sf RW}}
\label{Algorithm:Algorithm2}
\end{algorithm}

\begin{lemma} \label{lemma2:computation}
Given $D^{\ast}$ of an $(S,f)$-proper directed graph $D=(V,A,f,S)$, Algorithm ${\sf RW}$ correctly 
finds a $D^{\ast}$-rooted connector whose size is equal to ${\sf opt}_D$
in $O(\gamma_2 + M|A|)$ time if one exists where 
$\gamma_2$ is the time required to solve $\mbox{WMI}(D_+^{\ast})$ and 
$M=\sum_{v \in V}f(R_D(v))$.
\end{lemma}
\begin{proof}
The correctness of the algorithm follows from Lemma~\ref{lemma2:wmi}. We consider the time 
complexity. In Step~1, we can construct $D^{\ast}_+$ in $O(M|A|)$ time 
since $D_+^{\ast}$ has ${\sf opt}_D$ arcs 
parallel to each arc in $A$ and from (\ref{eq1:rooted}). 
Hence, since the time required for Step~2 is equal to $\gamma_2$,
the lemma holds.  
\end{proof}

\subsection{Algorithm for $\mbox{CDGI}$}

We are ready to explain the formal description of our algorithm 
called Algorithm {\sf Covering} for $\mbox{CDGI}(D)$.
Algorithm {\sf Covering} is the same as Algorithm {\sf CR} such that 
Steps~4, 5 and 6 are replaced by Algorithm {\sf RW}.

\begin{theorem}
Given a directed graph  $D=(V,A,S,f)$, Algorithm {\sf Covering} correctly finds 
a $D$-canonical set of in-trees which covers $A$ in $O(M^7|A|^6)$ time 
if one exits where $M=\sum_{v\in V}f(R_D(v))$.
\end{theorem}
\begin{proof}
The correctness of the algorithm follows from Lemmas~\ref{lemma1:computation} and 
\ref{lemma2:computation}. We then consider the time complexity of this algorithm.
From Lemmas~\ref{lemma1:computation} and \ref{lemma2:computation}, 
what remains is to analyze the time required to solve $\mbox{WMI}(D_+^{\ast})$. 
If $D$ is $(S,f)$-proper,  
$|A^{\ast}|
=\mbox{$\sum$}_{v \in V}|\delta_{D^{\ast}}(v)| 
\le \mbox{$\sum$}_{v\in V}f(R_D(v))=M$.  
Thus, since $D_+^{\ast}$ has ${\sf opt}_D$ parallel arcs of every $e \in A$,  
$|A_+^{\ast}| = |A^{\ast}| + \mbox{$\sum$}_{e \in A}{\sf opt}_D\le M + M|A|$.
Hence we have $|A_+^{\ast}|=O(M|A|)$. 
Thus, from Lemma~\ref{lemma:matroid}, we can solve $\mbox{WMI}(D^{\ast})$ in 
$O(M^7|A|^6)$ time. 
From this discussion and Lemmas~\ref{lemma1:computation} 
and \ref{lemma2:computation}, we obtain the theorem.  
\end{proof}

\section{Acyclic Case}
\label{Acyclic Case}

In this section, we show that in the case where $D=(V,A,S,f)$ is acyclic, 
a $D$-canonical set of in-trees covering $A$ 
can be computed more efficiently than the general case. 
For this, we prove the following theorem. 

\begin{theorem} \label{main_theorem_acyclic}
Given an acyclic directed graph  $D=(V,A,S,f)$, 
there exists a $D$-canonical set of in-trees which covers $A$
 if and only if 
\begin{equation} \label{eq1:main_theorem_acyclic}
|B| \le f(R_D(\partial^+(B))) \mbox{ for every } v \in V  \mbox{ and
 } B\subseteq \delta_D(v). 
\end{equation}
\end{theorem}
\begin{proof}
For each $v \in V$, we define an undirected bipartite graph
$G_v=(X_v\cup Y_v, E_v)$ which is necessary to prove the theorem. 
Let $X_v=\{x_e\colon e \in \delta_D(v)\}$ and 
$Y_v=\{y_{i,j}\colon s_i \in R_D(v), j=1,\ldots,f(s_i)\}$. 
$x_e \in X_v$ and $y_{i,j} \in Y_v$ are connected by an edge in $E_v$ if and only if $s_i$
is reachable from $\partial^+(e)$ (see Figure~\ref{fig:acyclic}). 

\begin{figure}[h]
\begin{minipage}{0.5\hsize}
\begin{center}
\includegraphics[width=4cm]{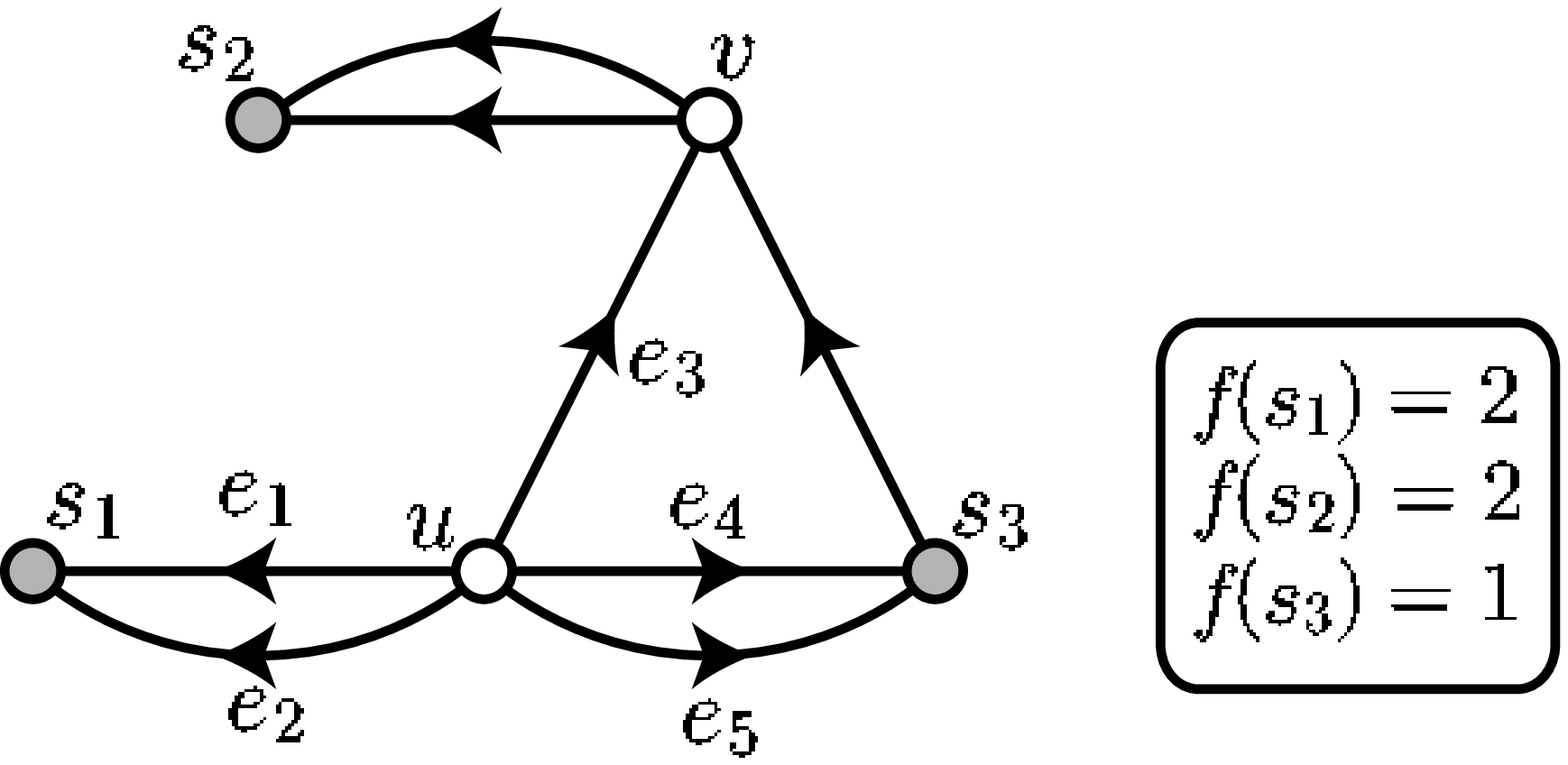}
\par(a)
\end{center}
\end{minipage}
\begin{minipage}{0.5\hsize}
\begin{center}
\includegraphics[width=3cm]{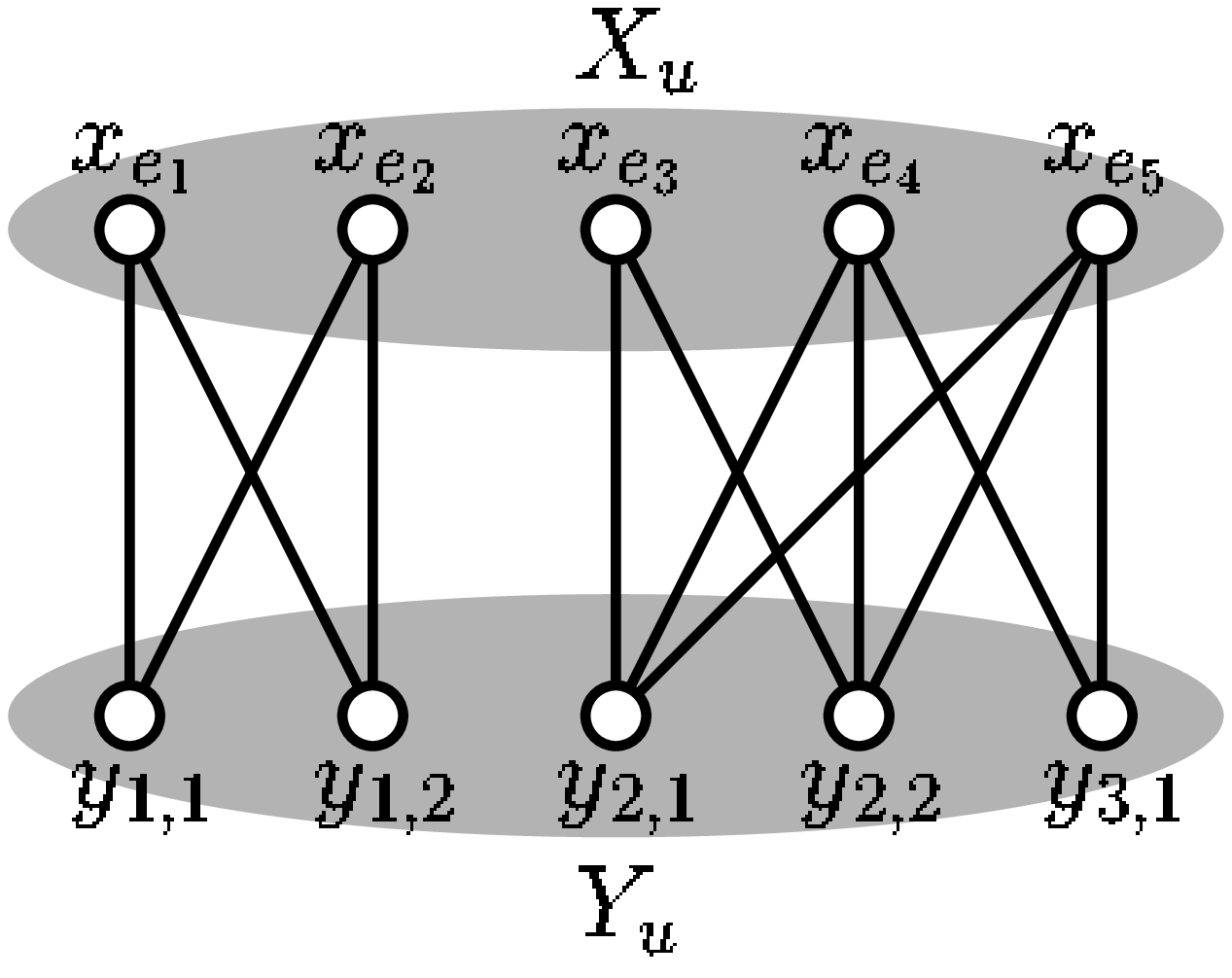}
\par(b)
\end{center}
\end{minipage}
\caption{\small (a) Input acyclic directed graph $D$. (b) Bipartite graph $G_u$
 for $u$ in (a).}
\label{fig:acyclic}
\end{figure}

It is well-known that (\ref{eq1:main_theorem_acyclic}) is equivalent to the necessary and
 sufficient condition that for any $v \in V$, there exists a matching
 in $G_v$ which saturates vertices in $X_v$ (e.g., Theorem~16.7 in Chapter~16 of \cite{S03}). 
Thus it is
 sufficient to prove that there exists a $D$-canonical set of in-trees which
 covers $A$ if and only if for any $v \in V$, there exists a matching
 in $G_v$ which saturates vertices in $X_v$. \\
{\bf If-part$\colon$}
Since $D$ has no cycle, we can label vertices in $V$ as follows, based on 
topological ordering$\colon$ 
(i) A label of each vertex is an integer between $1$ and $|V|$.
(ii) For any $e\in A$, a label of $\partial^+(e)$ is smaller than that of $\partial^-(e)$.  
For $W\subseteq V$, we denote by $D[W]$ a subgraph of 
$D=(V,A,S,f)$ induced by $W$ with a set of 
specified vertices $S\cap W$ and a restriction of $f$ on $S\cap W$. 
Let $V_t$ be the set of all vertices whose label is at most
$t$. We prove by induction on $t$. For $t=1$, 
it is clear that there exists a $D[V_1]$-canonical 
set of in-trees covering the arc set of $D[V_1]$. 
Assume that in the case of $t\ge 1$, there exists a $D[V_t]$-canonical set $\mathcal{T}$ 
of in-trees covering the arc set of $D[V_t]$. 
For $s_i \in S\cap V_t$ and $j=1,\ldots,f(s_i)$, 
let $T_{i,j}$ be an in-tree of $\mathcal{T}$ 
which is rooted at $s_i$ and spans vertices in $V_t$ from which $s_i$ is reachable. 

Let $v$ be a vertex whose label is equal to $t+1$.\\ 
{\bf Case1$\colon$}
We first consider the case of $v \notin S$. 
In this case, from $S\cap V_t=S\cap V_{t+1}$, 
we will construct a set $\mathcal{T}'$ of in-trees which consists of $T'_{i,1},\ldots,T_{i,f(s_i)}'$ 
for $s_i \in S\cap V_{t}$ ($=S\cap V_{t+1}$) such that each $T'_{i,j}$ 
is obtained from $T_{i,j}$. We first consider $T'_{i,j}$ for $s_i \in (S\cap V_{t})\setminus R_D(v)$. 
For $s_i \in (S\cap V_t)\setminus R_D(v)$, 
from $V^i_{D[V_t]}=V^i_{D[V_{t+1}]}$ holds, 
$T_{i,j}$ is also a $(D[V_{t+1}],s_i)$-in-tree. Thus, we set $T'_{i,j}=T_{i,j}$. 
Next we consider $T'_{i,j}$ for $s_i \in R_D(v)$. 
For $s_i \in R_D(v)$, since $V^i_{D[V_{t+1}]}=V^i_{D[V_{t}]}\cup \{v\}$ holds,
we need to add an arc in $\delta_{D}(v)$ to $T_{i,j}$. 
Here we use a matching $\mathcal{M}$ in $G_v$ which saturates vertices in $X_v$. 
For each edge $x_ey_{i,j} \in \mathcal{M}$, we set $T'_{i,j}$ be an in-tree obtained by adding 
an arc $e$ to $T_{i,j}$. If there exists $y_{i',j'} \in Y_v$  
which is not contained in any edge in $\mathcal{M}$, 
we arbitrarily choose an arc $e'\in \delta_D(v)$ such that $x_{e'}$ is a neighbour of $y_{i',j'}$ in $G_v$
and we set $T'_{i',j'}$ to be 
an in-tree obtained by adding $e'$ to $T'_{i',j'}$. 
From the way of construction, $\mathcal{T}'$ is clearly a $D[V_{t+1}]$-canonical set of in-trees.
Since $M$ saturates vertices in $X_v$, $T'_{i,1},\ldots,T'_{i,f(s_i)}$ with $s_i \in R_D(v)$ contain all arcs in 
$\delta_D(v)$. Thus, since $\mathcal{T}$ covers the arc set of $D[V_t]$ from the induction hypothesis, 
$\mathcal{T}'$ covers the arc set of $D[V_{t+1}]$. \\
{\bf Case2$\colon$}
Next we consider the case of $v \in S$. In this case, since $(S\cap V_{t})\setminus (S\cap V_{t+1})=\{v\}$ holds, 
letting $v =s_i$, 
we need to add new in-trees $T'_{i,j}=(\{s_i\}, \emptyset)$ for every $j=1,\ldots,f(s_i)$ 
to $\mathcal{T}'$ which is constructed as above. 
This completes the proof of the ``if-part''.\\
{\bf Only if-part$\colon$}
Assume that there exists a $D$-canonical set $\mathcal{T}$ of in-trees covering
$A$. For $i=1,\ldots,d$, we denote $f(s_i)$ $(D,s_i)$-in-trees
of $\mathcal{T}$ by $T_{i,1},\ldots,T_{i,f(s_i)}$.
Let us fix $v \in V$, and for $X_v$ and $Y_v$ we define a set $E'$ in which
an edge $x_e y_{i,j}$ is contained in $E'$ if and only if $e \in
\delta_D(v)$ is contained in $T_{i,j}$. If $e\in
\delta_D(v)$ is contained in $T_{i,j}$, $s_i$
is reachable from $\partial^+(e)$. Thus, 
$E'$ is a subset of $E_v$. 
Since $\mathcal{T}$ covers $A$,
each $e \in \delta_D(v)$ is contained in at least one in-tree in
$\mathcal{T}$. That is, $E'$ saturates $X_v$. 
Since $T_{i,j}$ is an in-tree, 
each $y_{i,j}$ is contained in exactly one edge in $E'$. 
Thus, it is not difficult to see
that a matching in $G_v$ which saturates vertices in $X_v$ can be
obtained from $E'$. This completes the proof.  
\end{proof}
From Theorem~\ref{main_theorem_acyclic}, 
instead of the algorithm presented in Section~\ref{Algorithm}, 
we can more efficiently find a $D$-canonical set
of in-trees covering $A$ by finding a maximum matching in
a bipartite graph $O(|V|)$ times.  
In regard to algorithms for finding a maximum matching in a bipartite
graph, see e.g.~\cite{HK73}. 

\begin{corollary} 
Given an acyclic directed graph  $D=(V,A,S,f)$, 
we can find a $D$-canonical set of in-trees which covers $A$
in $O({\sf match}(M+|A|, M|A|))$ time if one exists 
where ${\sf match}(n,m)$ represents the time 
required to find maximum matching in a bipartite graph with $n$ vertices and 
$m$ arcs and $M=\sum_{v \in V}f(R_D(v))$. 
\end{corollary}
\begin{proof}
From the proof of Theorem~\ref{main_theorem_acyclic}, for each $v \in
V$, $|X_v|=|\delta_D(v)|$ and $|Y_v|=f(R_D(v))$ hold. Then,
 $|E_v|=O(|\delta_D(v)|\cdot f(R_D(v)))$ follows. 
Thus, the corollary follows from 
$\sum_{v \in V}(|X_v|+|Y_v|)=M+|A|$ and $\sum_{v \in V}|E_v|=M|A|$. 
\end{proof}
{\bf Acknowledgement$\colon$}
We thank Prof.~Tibor Jord\'{a}n who informed us of the paper \cite{F06} and 
we are grateful to Shin-ichi Tanigawa for helpful comments.

\end{document}